%% file: data_words_synthesis.tex
\documentclass{article}
\usepackage[T1]{fontenc}
\usepackage{amsthm}
\usepackage{graphicx}
\input{macros}

\usepackage[a4paper, total={6in, 8in}]{geometry}

\newtheorem{theorem}{Theorem}
\newtheorem{lemma}[theorem]{Lemma}

\begin{document}
\title{First order synthesis for data words revisited}
%
%
\author{Julien Grange\footnote{Univ Paris Est Creteil, LACL, F-94010 Creteil, France}\and
Mathieu Lehaut\footnote{University of Gothenburg, Sweden}}
\date{}

%
%
%
\maketitle              

\begin{abstract}
  We carry on the study of the synthesis problem on data words for fragments of first order logic, and delineate precisely the border between decidability and undecidability.  
\end{abstract}

\section{Introduction}
\input{introduction}

\section{Preliminaries}
\input{definitions}

\section{\FOeq with processes partitioned between players}
\input{equivalence}

\section{Undecidability results}
\input{undecidability}

\section{Conclusion}
\input{conclusion}

%
%
\bibliographystyle{splncs04}
\bibliography{biblio}

\end{document}

%% file: macros.tex
\usepackage{mathtools}
\usepackage{amsmath}
\usepackage{amssymb}
\usepackage{algorithm}
\usepackage{algpseudocode}
\usepackage{enumitem}
\usepackage{xparse}
\usepackage{enumerate}
\usepackage{xspace}
\usepackage{algorithm}
\usepackage{algpseudocode}
\usepackage{float}
\usepackage{xcolor}
\usepackage{thm-restate}
\usepackage{pdfpages}
\usepackage{hyperref}
\usepackage{tikz}
\usetikzlibrary{patterns}
\usetikzlibrary{fit}
\usetikzlibrary{calc}
\usetikzlibrary{decorations.pathreplacing}
\usetikzlibrary{decorations.text}

\definecolor{mygreen}{RGB}{36, 200, 100}
\definecolor{myyellow}{RGB}{220, 220, 30}

\newcommand{\mathcommand}[2]{\newcommand{#1}{\ensuremath{#2}\xspace}}

\mathcommand{\FO}{\textsc{FO}}
\mathcommand{\Var}{V}
\mathcommand{\FOtwo}{\FO^2}
\mathcommand{\FOthree}{\FO^3}
\mathcommand{\MSO}{\textsc{MSO}}
\NewDocumentCommand{\FOeq}{O{}}{\ensuremath{\FO_{#1}[\sim]}\xspace}
\NewDocumentCommand{\FOtwoeq}{O{}}{\ensuremath{\FOtwo_{#1}[\sim]}\xspace}
\NewDocumentCommand{\FOtwoall}{}{\ensuremath{\FO^{2}[\sim,<,\plusone]}\xspace}
\NewDocumentCommand{\FOall}{}{\ensuremath{\FO[\sim,<,\plusone]}\xspace}
\NewDocumentCommand{\FOtwoeqord}{}{\ensuremath{\FO^{2}[\sim,<]}\xspace}
\NewDocumentCommand{\FOtwoeqplusone}{}{\ensuremath{\FO^{2}[\sim,\plusone]}\xspace}
\NewDocumentCommand{\FOthreeeqplusone}{}{\ensuremath{\FO^{3}[\sim,\plusone]}\xspace}
\NewDocumentCommand{\MSOordeq}{O{}}{\ensuremath{\MSO^{#1}[\lesssim]}\xspace}

\mathcommand{\plusone}{+1}

\NewDocumentCommand{\PARTSYNTH}{O{\FOeq}}{\ensuremath{\textsc{PartSynth}(#1)}\xspace}
\NewDocumentCommand{\SHAREDSYNTH}{O{\FOeq}}{\ensuremath{\textsc{SharedSynth}(#1)}\xspace}

\NewDocumentCommand{\foeq}{O{k}}{\ensuremath{\equiv^{\FO}_{#1}}\xspace}
\NewDocumentCommand{\fotwoeq}{O{k}}{\ensuremath{\equiv^{\FOtwo}_{#1}}\xspace}   
\NewDocumentCommand{\msoeq}{O{k}}{\ensuremath{\equiv^{\MSO}_{#1}}\xspace}
\NewDocumentCommand{\msoordeqeq}{O{k}}{\ensuremath{\equiv^{\MSOordeq}_{#1}}\xspace}

\mathcommand{\PTIME}{\textsc{PTime}}
\mathcommand{\LOGSPACE}{\textsc{LogSpace}}
\mathcommand{\PSPACE}{\textsc{PSpace}}
\mathcommand{\ACO}{\textsc{AC}\ensuremath{^0}}

\mathcommand{\Alp}{A}
\mathcommand{\Proc}{\mathbb{P}}

\newcommand{\Sys}{System\xspace}
\newcommand{\Envi}{Environment\xspace}

\mathcommand{\sig}{A}
\mathcommand{\sigE}{\sig_E}
\mathcommand{\sigS}{\sig_S}

\newcommand{\sys}{\mathrm{S}}
\newcommand{\env}{\mathrm{E}}
\mathcommand{\AlpE}{\Alp_\env}
\mathcommand{\AlpS}{\Alp_\sys}
\mathcommand{\ProcE}{\Proc_\env}
\mathcommand{\ProcS}{\Proc_\sys}
\mathcommand{\ProcM}{\Proc_{\sys\env}}

\NewDocumentCommand{\minind}{O{\formule}}{\ensuremath{n^\text{min}_\text{E}}\xspace}
\NewDocumentCommand{\cut}{O{\formule}}{\ensuremath{n^\text{min}_\text{S}}\xspace}
\mathcommand{\config}{C}
\mathcommand{\configinit}{C_0}
\mathcommand{\pos}{l}
\mathcommand{\posinit}{l_0}
\mathcommand{\posbis}{\pos'}
\mathcommand{\spel}{\hat \pos}
\NewDocumentCommand{\afterE}{O{\pos}}{E^*(#1)}
\NewDocumentCommand{\succE}{O{\pos}}{E^{+1}(#1)}
\mathcommand{\conf}{C}
\mathcommand{\confbis}{\widehat\conf}
\mathcommand{\confS}{\conf_\text{S}}
\mathcommand{\confbisS}{\confbis_\text{S}}
\mathcommand{\confE}{\conf_\text{E}}
\mathcommand{\confbisE}{\confbis_\text{E}}
\NewDocumentCommand{\numS}{O{\pos}}{\ensuremath{\confS(#1)}\xspace}
\NewDocumentCommand{\numE}{O{\pos}}{\ensuremath{\confE(#1)}\xspace}
\NewDocumentCommand{\numafterE}{O{\pos}}{\confE^*(#1)}
\NewDocumentCommand{\pot}{O{\pos}}{\rho(#1)}
\mathcommand{\victory}{\mathcal F}
\mathcommand{\game}{\mathfrak G_{\formule}}

\mathcommand{\vocabE}{\vocab_E}
\mathcommand{\vocabS}{\vocab_S}

\newcommand{\nE}{\ensuremath{n_E}\xspace}
\newcommand{\nS}{\ensuremath{n_S}\xspace}

\NewDocumentCommand{\fE}{O{k}}{\ensuremath{f_E(#1)}\xspace}
\NewDocumentCommand{\fS}{O{k}O{\nE}}{\ensuremath{f_S(#1,#2)}\xspace}

\mathcommand{\States}{\mathcal Q}
\mathcommand{\state}{q}
\mathcommand{\statebis}{\state'} 
\mathcommand{\istate}{\state_0}
\mathcommand{\stateone}{\state_1}
\mathcommand{\statetwo}{\state_2}
\mathcommand{\hstate}{\state_h}
\mathcommand{\Trans}{\mathcal T}
\mathcommand{\trans}{t}
\mathcommand{\transzero}{\trans_0}
\mathcommand{\transone}{\trans_1}
\mathcommand{\transtwo}{\trans_2}
\mathcommand{\transthree}{\trans_3}
\mathcommand{\mm}{M}
\mathcommand{\formmm}{\formule_{\mm}}

\NewDocumentCommand{\transition}{O{\trans}O{\state}O{\state}O{\cnt++}}{\ensuremath{#1:#2\xrightarrow{#4}#3}\xspace}

\newcommand{\strinc}{++}
\newcommand{\strdec}{--}
\newcommand{\strzero}{==0}

\NewDocumentCommand{\Transinc}{O{i}}{\ensuremath{\Trans^{\strinc}_{#1}}\xspace}
\NewDocumentCommand{\Transdec}{O{i}}{\ensuremath{\Trans^{\strdec}_{#1}}\xspace}
\NewDocumentCommand{\Transzero}{O{i}}{\ensuremath{\Trans^{\strzero}_{#1}}\xspace}

\NewDocumentCommand{\good}{O{}O{}O{x}}{\ensuremath{\text{Valid}_{#1}^{#2}(#3)}\xspace}
\NewDocumentCommand{\goode}{O{}O{x}}{\good[E][#1][#2]}
\NewDocumentCommand{\goods}{O{}O{x}}{\good[S][#1][#2]}
\NewDocumentCommand{\goodeinc}{O{x}}{\goode[\strinc][#1]}
\NewDocumentCommand{\goodsinc}{O{x}}{\goods[\strinc][#1]}
\NewDocumentCommand{\goodedec}{O{x}}{\goode[\strdec][#1]}
\NewDocumentCommand{\goodsdec}{O{x}}{\goods[\strdec][#1]}
\NewDocumentCommand{\goodezero}{O{x}}{\goode[\strzero][#1]}
\NewDocumentCommand{\goodszero}{O{x}}{\goods[\strzero][#1]}

\mathcommand{\ko}{\text{ko}}
\mathcommand{\ok}{\text{ok}}
\mathcommand{\noop}{\text{noop}}
\mathcommand{\koe}{\ko_E}
\mathcommand{\oke}{\ok_E}
\mathcommand{\kos}{\ko_S}
\mathcommand{\oks}{\ok_S}

\newlist{enuE}{enumerate}{1}
\setlist[enuE]{label=\textbf{[E\arabic*]}}
\newlist{enuS}{enumerate}{1}
\setlist[enuS]{label=\textbf{[S\arabic*]}}

\NewDocumentCommand{\ce}{O{i}}{\ensuremath{e_{#1}}\xspace}
\NewDocumentCommand{\cs}{O{i}}{\ensuremath{s_{#1}}\xspace}
\NewDocumentCommand{\cnt}{O{i}}{\ensuremath{c_{#1}}\xspace}
\mathcommand{\cntzero}{\cnt[0]}
\mathcommand{\cntone}{\cnt[1]}

\mathcommand{\eproc}{\bullet}

\NewDocumentCommand{\inc}{O{i}}{\ensuremath{\text{inc}_{#1}}\xspace}
\NewDocumentCommand{\dec}{O{i}}{\ensuremath{\text{dec}_{#1}}\xspace}

\newcommand{\sstyle}[1]{\textcolor{blue}{#1}}
\newcommand{\estyle}[1]{\textcolor{red}{#1}}

\mathcommand{\formkoe}{\Phi_{\koe}}
\mathcommand{\formkos}{\Phi_{\kos}}
\NewDocumentCommand{\formbadseq}{O{x}}{\ensuremath{\Psi_{\text{bad sequence}}(#1)}\xspace}
\NewDocumentCommand{\formbadtarget}{O{x}}{\ensuremath{\Psi_{\text{bad target}}(#1)}\xspace}
\NewDocumentCommand{\formbadsource}{O{x}}{\ensuremath{\Psi_{\text{bad source}}(#1)}\xspace}
\NewDocumentCommand{\formbadupkeep}{O{x}}{\ensuremath{\Psi_{\text{bad upkeep}}(#1)}\xspace}
\NewDocumentCommand{\formbadzerotest}{O{x}}{\ensuremath{\Psi_{\text{bad zero test}}(#1)}\xspace}

\NewDocumentCommand{\isstate}{O{x}}{\ensuremath{\States(#1)}\xspace}
\NewDocumentCommand{\istrans}{O{x}}{\ensuremath{\Trans(#1)}\xspace}
\NewDocumentCommand{\istransinc}{O{x}}{\ensuremath{\Transinc(#1)}\xspace}
\NewDocumentCommand{\istransdec}{O{x}}{\ensuremath{\Transdec(#1)}\xspace}
\NewDocumentCommand{\istranszero}{O{x}}{\ensuremath{\Transzero(#1)}\xspace}
\NewDocumentCommand{\isupkeep}{O{x}}{\ensuremath{\mathcal U(#1)}\xspace}
\NewDocumentCommand{\issys}{O{x}}{\ensuremath{\sigS(#1)}\xspace}
\NewDocumentCommand{\isenv}{O{x}}{\ensuremath{\sigE(#1)}\xspace}

\mathcommand{\formprefixE}{\Phi_{\text{bad prefix E}}}
\mathcommand{\formprefixS}{\Phi_{\text{bad prefix S}}}
\mathcommand{\formblockE}{\Phi_{\text{E blocks}}}
\mathcommand{\formblockS}{\Phi_{\text{S blocks}}}
\mathcommand{\formplayafterkoE}{\Phi_{\text{E plays after ko}}}
\mathcommand{\formplayafterkoS}{\Phi_{\text{S plays after ko}}}

\mathcommand{\first}{\texttt{first}}
\mathcommand{\second}{\texttt{second}}
\mathcommand{\last}{\texttt{last}}

\mathcommand{\N}{\mathbb{N}}
\mathcommand{\Zero}{\mathbb O}
\mathcommand{\im}{\operatorname{Im}}

\newcommand{\contrib}{\textbf}

\mathcommand{\logic}{\mathcal L}
\mathcommand{\logicbis}{\logic'}

\mathcommand{\thr}{t}

\mathcommand{\elem}{a}
\mathcommand{\elembis}{b}
\mathcommand{\elemter}{c}

\mathcommand{\var}{x}
\mathcommand{\varbis}{y}

\mathcommand{\classe}{\mathcal C}
\mathcommand{\prop}{\mathcal P}

\mathcommand{\structdom}{A}
\mathcommand{\struct}{\mathcal A}
\mathcommand{\structbisdom}{A'}
\mathcommand{\structbis}{\mathcal A'}

\mathcommand{\type}{\tau}
\mathcommand{\typebis}{\type'}

\mathcommand{\formule}{\varphi}
\mathcommand{\formulebis}{\psi}
\mathcommand{\formuleter}{\theta}
\mathcommand{\vocab}{\Sigma}

\mathcommand{\graphe}{\mathcal\graphedom}
\mathcommand{\graphebis}{\mathcal\graphebisdom}
\mathcommand{\graphedom}{G}
\mathcommand{\graphebisdom}{\graphedom'}
\mathcommand{\graphedombis}{H}
\mathcommand{\edgerel}{E}

\mathcommand{\motvide}{\varepsilon}

\NewDocumentCommand{\initsegment}{O{n}}{\ensuremath{[1,#1]}\xspace}
\mathcommand{\oldsysproc}{\initsegment[\nS+1]}
\mathcommand{\newsysproc}{\initsegment[\nS]}
\mathcommand{\envproc}{\initsegment[\nE]}

\mathcommand{\oldProcS}{\ProcS^{\nS+1}}
\mathcommand{\newProcS}{\ProcS^{\nS}}

\NewDocumentCommand{\ps}{O{}O{}}{\ensuremath{s_{#1}^{#2}}\xspace}
\NewDocumentCommand{\oldps}{O{}}{\ps[#1][\nS+1]}
\NewDocumentCommand{\newps}{O{}}{\ps[#1][\nS]}
\NewDocumentCommand{\pe}{O{}}{\ensuremath{e_{#1}}\xspace}

\NewDocumentCommand{\strat}{O{}}{\ensuremath{\mathcal S_{#1}}\xspace}
\mathcommand{\oldstrat}{\strat[\nS+1]}
\mathcommand{\newstrat}{\strat[\nS]}

\NewDocumentCommand{\map}{O{}}{\ensuremath{\sigma_{#1}}\xspace}
\mathcommand{\prevmap}{\map[r]}
\mathcommand{\nextmap}{\map[r+1]}

\NewDocumentCommand{\ghost}{O{}}{\ensuremath{\gamma_{#1}}\xspace}
\mathcommand{\prevghost}{\ghost[r]}
\mathcommand{\nextghost}{\ghost[r+1]}      

\NewDocumentCommand{\actarray}{O{}}{\ensuremath{\texttt{act}_{#1}}\xspace}
\mathcommand{\prevactarray}{\actarray[r]}
\mathcommand{\nextactarray}{\actarray[r+1]}
\NewDocumentCommand{\act}{O{}O{\ps}}{\ensuremath{\actarray[#1][#2]}\xspace}
\NewDocumentCommand{\prevact}{O{\ps}}{\ensuremath{\act[r][#1]}\xspace}
\NewDocumentCommand{\nextact}{O{\ps}}{\ensuremath{\act[r+1][#1]}\xspace}

\NewDocumentCommand{\invtp}{O{}O{}}{\ensuremath{(\texttt{TP}^\texttt{#2}_{#1})}\xspace}
\NewDocumentCommand{\invstp}{O{}}{\invtp[#1][S]}
\mathcommand{\previnvstp}{\invstp[r]}
\mathcommand{\nextinvstp}{\invstp[r+1]}
\NewDocumentCommand{\invetp}{O{}}{\invtp[#1][E]}
\mathcommand{\previnvetp}{\invetp[r]}
\mathcommand{\nextinvetp}{\invetp[r+1]}

\NewDocumentCommand{\invact}{O{}}{\ensuremath{(\texttt{Act}_{#1})}\xspace}
\mathcommand{\previnvact}{\invact[r]}
\mathcommand{\nextinvact}{\invact[r+1]}

\NewDocumentCommand{\invghosttp}{O{}}{\ensuremath{(\ghost\texttt{TP}_{#1})}\xspace}
\mathcommand{\previnvghosttp}{\invghosttp[r]}
\mathcommand{\nextinvghosttp}{\invghosttp[r+1]}

\NewDocumentCommand{\invghostcc}{O{}}{\ensuremath{(\ghost\texttt{CC}_{#1})}\xspace}
\mathcommand{\previnvghostcc}{\invghostcc[r]}
\mathcommand{\nextinvghostcc}{\invghostcc[r+1]}

\NewDocumentCommand{\invmax}{O{}}{\ensuremath{(\texttt{Max}_{#1})}\xspace}
\mathcommand{\previnvmax}{\invmax[r]}
\mathcommand{\nextinvmax}{\invmax[r+1]}

\NewDocumentCommand{\dw}{O{}O{}}{\ensuremath{w_{#1}^{#2}}\xspace}
\mathcommand{\prevolddw}{\dw[r][\nS+1]}
\mathcommand{\nextolddw}{\dw[r+1][\nS+1]}
\mathcommand{\prevnewdw}{\dw[r][\nS]}
\mathcommand{\nextnewdw}{\dw[r+1][\nS]}

\NewDocumentCommand{\msotp}{O{p}O{\dw}O{k}}{\ensuremath{\text{mso-tp}^{#3}(#1\,|\,#2)}\xspace}
\NewDocumentCommand{\prevoldmsotp}{O{\prevmap(s)}O{k}}{\msotp[#1][\prevolddw][#2]}
\NewDocumentCommand{\nextoldmsotp}{O{\nextmap(s)}O{k}}{\msotp[#1][\nextolddw][#2]}
\NewDocumentCommand{\prevnewmsotp}{O{s}O{k}}{\msotp[#1][\prevnewdw][#2]}
\NewDocumentCommand{\nextnewmsotp}{O{s}O{k}}{\msotp[#1][\prevnewdw][#2]}

\NewDocumentCommand{\height}{O{\type}}{\ensuremath{h(#1)}\xspace}
\NewDocumentCommand{\boundcc}{O{\height}}{\ensuremath{F_\text{CC}(#1)}\xspace}
\NewDocumentCommand{\boundtp}{O{\height}}{\ensuremath{F_\text{TP}(#1)}\xspace}  

\NewDocumentCommand{\CC}{O{\type}}{\ensuremath{\text{CC}(#1)}\xspace}
\mathcommand{\prevghostcc}{\CC[\oldmsotp[\prevghost]]}
\mathcommand{\nextghostcc}{\CC[\oldmsotp[\nextghost]]}

\NewDocumentCommand{\tpord}{O{k}}{\ensuremath{\prec_{#1}}\xspace}

\NewDocumentCommand{\auto}{O{k}}{\ensuremath{\mathfrak A_{#1}}\xspace}
\NewDocumentCommand{\preordtp}{O{k}}{\ensuremath{\rightharpoonup_{#1}}\xspace}
\NewDocumentCommand{\equivtp}{O{k}}{\ensuremath{\rightleftharpoons_{#1}}\xspace}

\mathcommand{\false}{\texttt{false}}
\mathcommand{\true}{\texttt{true}}

\mathcommand{\seqE}{\dw[E]}
\mathcommand{\seqS}{\dw[S]}
\mathcommand{\longseq}{\widehat{\dw}}

\mathcommand{\resetact}{\texttt{reset\_act()}}
\mathcommand{\restr}{\upharpoonright}

%% file: introduction.tex
The reactive synthesis problem, which dates back to Church~\cite{sisl1957-Chu}, is about generating a \emph{correct-by-construction} program with respect to a given specification.
It is often formulated as a two-player game between an uncontrollable Environment and the System, who alternate picking an input and an output letter, respectively.
This creates an infinite execution, and the goal of the System is to make every execution satisfy the specification, whatever Environment does.
If the System has a strategy to ensure this result, it then corresponds to a program that is sure to respect the specification.
The original problem is decidable and was solved by B\"{u}chi and Landweber \cite{TAMS138-BL}, and several improvements and extensions have since been studied. However, it only encompasses finite alphabets, which is inadequate for representing executions of distributed systems involving a number of processes which is not fixed; this occurs in communication protocols, distributed algorithms, multi-agent systems, swarm robotics, or with ad-hoc networks.

We thus consider an extension of this problem which deals with alphabets whose size is not fixed. 
In those cases it is more adequate to use data words to represent executions, where a data word consists of a sequence of pairs (action, data), the action coming from a finite alphabet while the data, which comes from an infinite alphabet (or at least, an alphabet whose size is not fixed), specifies who performed said action.

We will assume that the sets of actions of both players are disjoint, but this does not necessarily need to be the case for the processes.
We will consider two distinct cases, depending on process ownership.
In the first case, all processes are shared, meaning that both players can perform their respective actions on any process.
Consider for instance the modelling of a drone fleet, where a process corresponds to a single drone, inputs correspond to atmospheric conditions and outputs to possible movements of the drone. In that case, it makes sense to have both \Sys and \Envi play their actions on shared processes.
In the second case, each player has their own processes, on which only they can play an action.
This is for instance useful for modelling a single machine with several components, where each component is considered to be its own process, some being sensors/receptors (belonging to \Envi) while other parts (\Sys's processes) perform the machine's output.

Finally, it remains to chose a formalism to express the specification which should be satisified by System.
As always, there is a trade off between expressiveness of the formalism and tractability of its synthesis problem.
Many specification languages for data words have been studied, and the synthesis problem has been investigated for some of them, such as register automata~\cite{KhalimovMB18} and the Logic of Repeating Values~\cite{FigueiraP18}.
Here we follow the steps of~\cite{DBLP:conf/fossacs/BerardBLS20} and consider first order logic \FO and its fragment \FOtwo where only two reusable variables are allowed.
First order logic for words is well understood, and its extension to data words is easy to define, which makes it a good candidate for our purpose.
On data words, \FO has acess to a binary predicate $\sim$ such that $x \sim y$ if positions $x$ and $y$ belong to the same process. On top of that, a unary predicate for each action marks the positions where this action has been played (as is usually the case for words), and we will study variations where relative positions will either be accessible via a binary predicate $<$, its successor relation \plusone, both or none of them.

The satisfiability problem, which can be seen as a very restricted synthesis problem in the case where Environment never acts, has been showed to be decidable in the two-variable fragment with $<$ and \plusone (\FOtwoall)~\cite{DBLP:conf/lics/BojanczykMSSD06}.
The synthesis problem itself has previously been studied in~\cite{DBLP:conf/fossacs/BerardBLS20} for various of fragments of \FO, with both decidable and undecidable results, depending on the variety of the considered fragment.
The aim of this paper is to extend those results and get a better understanding on where the border between decidability and undecidability lies exactly.


%% file: definitions.tex
\paragraph*{Data words.}
Let \Alp be a finite alphabet of \emph{actions}, partitioned into \Sys actions \AlpS and \Envi actions \AlpE. Let \ProcS (resp. \ProcE) be a finite set of \emph{processes} belonging to \Sys (resp. \Envi), and let \ProcM be a finite set of processes which can be activated by both players. These sets are assumed to be pairwise disjoint, and we define $\Proc:=\ProcS\cup\ProcE\cup\ProcM$. Let us refer to $\AlpS \times (\ProcS\cup\ProcM)$ as \vocabS and to $\AlpE\times(\ProcE\cup\ProcM)$ as \vocabE. A \emph{data word} is a finite or infinite sequence $w = (a_0, p_0) (a_1, p_1) \dots$ over $\vocab:=\vocabS\cup\vocabE$.

Intuitively, a data word represents the trace of an execution of a distributed system.
Each pair $(a,p)$ (called a \emph{position}) indicates that action $a$ has been performed by process $p$. Note that the above definition prevents a player from playing on processes owned by their opponent.

\paragraph*{First order logic.}
Let \Alp be a finite alphabet of actions and \Var be a set of variable names.
We define formulas of first order logic with $<$ and \plusone on data words (\FOall) as follows:
\begin{align*}
t :=~ &a(x) \mid x=y \mid x < y \mid x = y \plusone \mid x\sim y \mid \ProcS(x)\mid \ProcE(x) \mid \ProcM(x)\\
\formule :=~ &t \mid \neg \formule \mid \formule \land \formule \mid \formule \lor \formule \mid \exists x. \formule \mid \forall x. \formule 
\end{align*}
where $a \in \Alp$ and $x,y \in \Var$.

A data word $w = (a_0, p_0) (a_1, p_1) \dots$ is represented as a structure whose domain has an element for each position $(a_i,p_i)$, and an element for each process in \ProcS (resp. \ProcE, resp. \ProcM), which is in the interpretation of the unary predicate \ProcS (resp. \ProcE, resp. \ProcM). Note that this allows to quantify on processes on which no action has been played.
For every action $a\in\Alp$, the unary predicate $a$ is interpreted as the set of positions with action $a$. The binary predicate $<$ is interpreted as the linear order on positions and \plusone as its successor relation. Finally, $(a_i,p_i)\sim(a_j,p_j)$ holds if and only if $p_i=p_j$; on top of that, $\sim$ also relates each process $p$ to all positions $(a,p)$.
Beyond that, we consider the usual semantics for first order logic, e.g. see~\cite{DBLP:books/sp/Libkin04}.


We then define various fragments of \FOall depending on which binary predicates are allowed, and how many variables are available.
\FOall is the whole fragment with $\sim$, $<$, and $\plusone$, while \FOeq only allows $\sim$ out of the three predicates.
Furthermore, we also study the two-variable restriction of \FO, in which one can only use two (reusable) variables in the whole formula, i.e. $\Var = \{x,y\}$.
We denote this variation by \FOtwo, and combine it with the previous notations; for instance \FOtwoeqord is the two-variable fragment with predicates $\sim$ and $<$, but without $\plusone$. 
Note that unless restricted to two variables, first order logic can define $\plusone$ using $<$, but not the other way around.


\paragraph*{Synthesis.}
Let us now define the synthesis problem. A data word is seen as an execution, or history, in a game opposing System and Environment.
A \emph{strategy} \strat for System is a function $\vocab^\ast \to \vocabS \cup \{\varepsilon\}$ which, given some finite data word representing the history of the game, either outputs some move in \vocabS (i.e. plays an action on some process which can be activated by \Sys) or passes its turn.
An execution $w = \sigma_0 \sigma_1 \dots$ is said to be $\strat$-compatible if
\begin{itemize}
\item for all $0\leq i < |w|$, $\sigma_i \in \vocabS$ implies $\sigma_i = \strat(\sigma_0 \dots \sigma_{i-1})$, and
\item if $w$ is finite, then $\strat(w) = \varepsilon$.
\end{itemize}
With this definition we intuitively allow Environment to block System from playing anytime Environments wants to play, which could potentially be forever.
To prevent pathological cases, we will consider only \emph{fair} executions for \strat, i.e. executions $w = \sigma_0 \sigma_1 \dots$ such that if there are infinitely many $i$ such that $\strat(\sigma_0 \dots \sigma_{i-1}) \neq \varepsilon$, then there are infinitely many positions in \vocabS. In other words, \Envi can delay actions from \Sys, but cannot silence them forever.
Finally, for a given formula $\formule$, we say that $\strat$ is \emph{winning} for \formule if all executions which are $\strat$-compatible and fair for \strat satisfy $\formule$. 

In this paper we focus on two specific configurations for processes: when they are shared, and when they are partitioned between players.
\paragraph*{Shared.} We say that the processes are \emph{shared} when $\ProcS = \ProcE = \emptyset$, i.e. when all processes can be affected both by System and Environment. The synthesis problem for logic \logic with shared processes is denoted \SHAREDSYNTH[\logic]: it amounts, given an alphabet \Alp of actions and a formula $\formule\in\logic$, to decide whether there exists a winning strategy \strat for \formule with $\ProcS = \ProcE = \emptyset$ and some (finite) $\ProcM$. Since the identity of the shared processes does not matter and only the cardinality of \ProcM is relevant, we say in that case that \strat is a $|\ProcM|$-winning strategy for \formule.
\paragraph*{Partitioned.} We say that the processes are \emph{partitioned} if $\ProcM = \emptyset$. In that case, each player has their own pool of processes on which they can play, but that their opponent cannot use. The synthesis problem for logic \logic with partitioned processes is denoted \PARTSYNTH[\logic]. As above, it is the problem of deciding, given \Alp and $\formule\in\logic$, whether \Sys has a winning strategy for \formule with $\ProcM=\emptyset$ and some arbitrary finite sets $\ProcS$ and $\ProcE$. Similarly, we say in that case that \strat is a $(|\ProcS|,|\ProcE|)$-winning strategy for \formule.

\paragraph*{Known results.}

Data words were introduced in~\cite{DBLP:conf/lics/BojanczykMSSD06}. Boja\'nczyk et al. proved that the satisfiability problem for \FOtwoall on data words is decidable. Note that this corresponds to the synthesis problem for \FOtwoall when both \ProcE and \ProcM are empty. They also showed that as soon as a third variable is available, this problem becomes undecidable, even without the order (i.e. for \FOthreeeqplusone).

Decidability of the satisfiability problem for the two-variable logic in this setting is what prompted Bérard et al. to consider the synthesis problem on data words, for serveral fragments of first order logic~\cite{DBLP:conf/fossacs/BerardBLS20}. They proved that the synthesis problem for \FOeq is decidable when processes are partitioned and when the number of \Envi processes is fixed. In constrast, they established undecidability results for \FOeq and \FOtwoall when processes are shared.

\paragraph{Contributions.}

We summarize the contributions of this paper in Table~\ref{tab:results}, in bold font. Results from~\cite{DBLP:conf/fossacs/BerardBLS20} are also mentioned. As can be seen, we bridge all the gaps left open by~\cite{DBLP:conf/fossacs/BerardBLS20}.

\begin{table}
  \def\arraystretch{1.2}
  \setlength\tabcolsep{1em}
  \begin{center}
    \begin{tabular}{| c | c | c |}
      \hline
      Logic $\backslash$ Processes & Partitioned & Shared\\
      \hline
      \FOtwoeq & \contrib{decidable} (Th~\ref{th:eq}) & \contrib{undecidable} (Th~\ref{th:twomixed})\\
      \FOeq & \contrib{decidable} (Th~\ref{th:eq}) & undecidable~\cite{DBLP:conf/fossacs/BerardBLS20}\\
      \FOtwoeqord & \contrib{undecidable} (Th~\ref{th:twoeqord}) & \contrib{undecidable} (Th~\ref{th:twomixed})\\
      \FOtwoeqplusone & \contrib{undecidable} (Th~\ref{th:twoeqplusone}) & \contrib{undecidable} (Th~\ref{th:twomixed})\\
      \FOtwoall & \contrib{undecidable} (Th~\ref{th:twoeqord} or \ref{th:twoeqplusone}) & undecidable~\cite{DBLP:conf/fossacs/BerardBLS20}\\
      \hline
    \end{tabular}
  \end{center}
  \caption{Decidability and undecidability of the synthesis problem for fragments of first order logic, in the pure and mixed cases. New results from this paper appear in bold.}
  \label{tab:results}
\end{table}

%% file: equivalence.tex
First, we turn to the case where processes are partitioned between \Sys and \Envi. It has been shown in~\cite{DBLP:conf/fossacs/BerardBLS20} that in that case, the synthesis problem for \FOeq is decidable when \Sys has an arbitrary number of processes, but \Envi only has access to a fixed number of processes. We extend this result by lifting this restriction: we show in Theorem~\ref{th:eq} that \PARTSYNTH[\FOeq] is decidable.

As a first step towards proving this result, let us show that beyond a certain threshold (which depends only on the formula), having access to more processes in $\ProcE$ is always a boon for the \Envi. Note that this is note true for small cardinalities of \ProcE: it is not hard to design a game where \Envi wins if $\ProcE=\emptyset$ but loses as soon as $|\ProcE|\geq 1$.

\begin{lemma}
  \label{lem:more_is_better}
  For every alphabet $A$ and every $\formule\in\FOeq$, there exists $\minind\in\N$ such that for every $\nS\in\N$ and every $\nE\geq\minind$, if \Sys has a $(\nS,\nE+1)$-winning strategy for $\formule$, then \Sys has an $(\nS,\nE)$-winning strategy for $\formule$.
\end{lemma}

\begin{proof}

  For this proof, we rely on a characterisation of the synthesis problem for \FOeq via parametrised vector games, defined in~\cite{DBLP:conf/fossacs/BerardBLS20}. The \emph{parametrised vector game} on alphabet $\Alp=\AlpS\uplus\AlpE$ with bound $B\in\N$, victory condition \victory and (\nS,\nE) pebbles is a game between two players, \Sys and \Envi, defined as follows. An \emph{S-location} (resp. $\emph{E-location}$) is a vector $\ll\! a_1^{\nu_1}\cdots a_d^{\nu_d}\!\gg$, where $a_1,\cdots,a_d$ are the letters from $\AlpS$ (resp. $\AlpE$), and $0\leq\nu_i\leq B$. An S-location (resp. E-location) $\ll\! a_1^{\mu_1}\cdots a_d^{\mu_d}\!\gg$ is reachable from the S-location (resp. E-location) $\ll\! a_1^{\nu_1}\cdots a_d^{\nu_d}\!\gg$ if $\mu_i\geq\nu_i$ for every $i$. The initial S-location (resp. E-location) is the one for which all $\nu_i$ are zero. \Sys (resp. \Envi) has a number \nS (resp. \nE) of pebbles, which are private to each player. The identity of the pebbles is irrelevant; only their number in each location matters. An S-configuration \confS (resp. an E-configuration \confE) is just a way for \Sys (resp. \Envi) to distribute their pebbles on their locations; namely, it is a function which maps every S-location (resp. E-location) \pos to $\numS\in\{0,\cdots,\nS\}$ (resp. ($\numE\in\{0,\cdots,\nE\}$) whose sum over the locations is \nS (resp. \nE). An S-configuration (and similarly for E-configurations) \confbisS is reachable from \confS if it is possible to obtain \confbisS from \confS by moving a number of pebbles from locations \pos to locations that are reachable from \pos.
  
  A \emph{location} is a couple composed of an S- and an E-locations, and the initial location \posinit is the couple of initial S- and E-locations. A \emph{configuration} $\conf=(\confS,\confE)$ is a couple composed of an S- and an E-configuration. In the initial configuration, the pebbles of both players are placed on their respective initial locations. An \emph{acceptance condition} $c$ is a mapping from the set of S- and E-locations to the set of conditions $\{``=n", ``\geq n" : n\in\N\}$, and a configuration $(\confS,\confE)$ satisfies this acceptance condition if for every S- or E-location \pos, \numS or \numE satisfies $c(\pos)$. The winning condition \victory is a finite set of acceptance conditions, and a configuration is said to be \emph{winning for \Sys} if it satisfies at least one acceptance condition in \victory; otherwise, it \emph{winning for \Envi}.

  A strategy for \Sys (resp. \Envi) is a function mapping a configuration $(\confS,\confE)$ to a configuration $(\confbisS,\confE)$ where \confbisS is reachable from \confS (resp. to a configuration $(\confS,\confbisE)$ where \confbisE is reachable from \confE). A strategy for \Sys (resp. \Envi) is said to be \emph{winning} if for every strategy for \Envi (resp. \Sys), the result of letting both strategies alternate moves against each other, starting with \Envi (which can easily be seen to lead to a well-defined final configuration) is winning for \Sys (resp. for \Envi). Parametrised vector games are determined.
  
  Parametrised vector games are shown in~\cite{DBLP:conf/fossacs/BerardBLS20} (Lemma 11) to be equivalent to the synthesis problem for \FOeq: for every formula $\formule\in\FOeq$, there exists a parametrised vector game \game such that for every $\nS,\nE\in\N$, \Sys has an $(\nS,\nE)$-winning strategy for \formule if and only if \Sys wins \game with $(\nS,\nE)$ pebbles.

  Thus, to prove the lemma, we show that if \Envi has a winning strategy in the parametrised vector game \game with $(\nS,\nE)$ pebbles, \Envi has a winning strategy in \game with $(\nS,\nE+1)$ pebbles, for every $\nE\geq\minind$.
  Let $\AlpE=\{a_1,\cdots,a_d\}$, $B$ be the bound of \game, and let $K$ be the largest constant that occurs in \victory.
  We claim that $\minind:=K\cdot(d+1)^{d\cdot B}$ fits the bill.

  Given an E-location $\pos:=\ll\! a_1^{\nu_1}\cdots a_d^{\nu_d}\!\gg$, let us denote by $\afterE$ the set of E-locations that can be reached from \pos, i.e. locations $\ll\! a_1^{\mu_1}\cdots a_d^{\mu_d}\!\gg$ for $\mu_i\geq\nu_i$, and by $\numafterE$ the number of \Envi pebbles in locations of $\afterE$, i.e. \[\numafterE:=\sum_{\posbis\in\afterE}\numE[\posbis]\,.\]
  Let $\succE$ be the subset of successors of $\pos$ in $\afterE$, i.e. locations $\ll\! a_1^{\mu_1}\cdots a_d^{\mu_d}\!\gg$ where $\mu_i=\nu_i$ except for one $i$, for which $\mu_i=\nu_i+1$.

  We define the potential $\pot$ of some E-location $\pos=\ll a_1^{\nu_1}\cdots a_d^{\nu_d}\gg$ as follows: $\pot:=\sum_{i=1}^d (B-\nu_i)$. In other words, $\pot$ is the number of times an \Envi pebble in $\pos$ can be moved. Note that the potention of the initial E-location is $d\cdot B$.

  Consider some E-configuration $\confE$, and let $P(\pos)$ denote the property
  \begin{equation}
    \numafterE\geq K\cdot (d+1)^{\pot}\,.\tag{$P(\pos)$}
  \end{equation}
  If $P(\pos)$, then either $\numE\geq K$ or there exists  $\posbis\in \succE$ such that $P(\posbis)$. 
  Indeed, note that \[\afterE=\{\pos\}\cup\bigcup_{\posbis\in \succE}\afterE[\posbis]\,.\]
  If $\pot=0$, there is nothing to prove. Otherwise, suppose that $\numE<K$ and $P(\posbis)$ fails for every $\posbis\in\succE$. Then we would have \[\numafterE\leq \numE + \sum_{\posbis\in\succE}\numafterE[\posbis]<K+d\cdot K(d+1)^{\pot-1}\leq K(d+1)^{\pot}\,,\] contradicting $P(\pos)$.

  With this in mind, a simple induction on $\pot$ establishes the following: if $P(\pos)$ holds for some $\pos$, there there exists some $\posbis\in\afterE$ such that $P(\posbis)$ and $\numE\geq K$.

  We are now ready to describe the winning strategy for \Envi in \game with $(\nS,\nE+1)$ pebbles, based on their winning strategy with $(\nS,\nE)$ pebbles: at all time, \Envi keeps a marked E-location $\spel$ where they place their additional pebble.

  We argue that it is at all time possible to choose an $\spel$ which is such that $P(\spel)$ and $\numE[\spel]\geq K$ in the configuration of the game with $(\nS,\nE)$ pebbles, such that \Envi can move the additional pebble from one $\spel$ to the next. Indeed, this holds in the initial configuration by choice of $\minind$ (in that case, $\spel$ is the initial E-location), and the previous observation ensures that it is always possible to find, after each move of \Envi in the original game, a new position $\pos$ of $\afterE[\spel]$, which satisfies $P(\pos)$ and $\numE\geq K$. Then we take such an $l$ as our new $\spel$.

  Given that $\numE[\spel]\geq K$ and that, by definition, $K$ is the largest bound appearing in the victory condition \victory, the additional pebble lying in $\spel$ is irrelevant to whether a congiguration is winning for \Envi or not.
\end{proof}

We now have the tools to prove our main decidability result:

\begin{theorem}
  \label{th:eq}
  \PARTSYNTH is decidable.
\end{theorem}

\begin{proof}
  Let $(A,\formule)$ be an input for \PARTSYNTH, and let \minind be the corresponding integer whose existence is guaranteed by Lemma~\ref{lem:more_is_better}.

  Theorem 15 from~\cite{DBLP:conf/fossacs/BerardBLS20} states that there exists some integer \cut such that \Sys has a $(\cut,\minind)$-winning strategy for \formule iff \Sys has a $(\nS,\minind)$-winning strategy for $\formule$ for any (every) $\nS\geq\cut$.

  We argue in the following that \Sys has an $(\nS,\nE)$-winning strategy for \formule for some $\nS,\nE\in\N$ if and only if \Sys has an $(\nS,\nE)$-winning strategy for $\formule$ for some $\nS\leq\cut$ and $\nE\leq\minind$. Thus there are only a finite number of process configurations to explore in order to decide whether \Sys has a $(\nS,\nE)$-winning strategy for $\formule$ for some $\nS,\nE\in\N$. When the numbers \nS and \nE of processes are fixed, one can duplicate the letter from the alphabet (for instance, multiplying \AlpS by the number of \Sys processes), and reduce the synthesis problem on data words to the decidable synthesis problem on plain words~\cite{TAMS138-BL}, and thus decide the problem in each of this finitely many fixed configurations. It is important to note that the bounds $\cut$ and $\minind$ are computable. 
  
  Let us now show that one need not look further than $(\cut,\minind)$: suppose that \Sys has an $(\nS,\nE)$-winning strategy for $\formule$. Let us consider four cases, which cover all the possible values of \nS and \nE.
  \begin{itemize}
  \item If $\nS\leq\cut$ and $\nE\leq\minind$, there is nothing to prove.
  \item If $\nS>\cut$ and $\nE>\minind$, then in virtue of Lemma~\ref{lem:more_is_better}, \Sys has an $(\nS,\minind)$-winning strategy for $\formule$. In turn, Theorem~15 from~\cite{DBLP:conf/fossacs/BerardBLS20} ensures \Sys has an $(\cut,\minind)$-winning strategy for $\formule$.
  \item If $\nS\leq\cut$ and $\nE>\minind$, then Lemma~\ref{lem:more_is_better} ensures \Sys has an $(\nS,\minind)$-winning strategy for $\formule$.
  \item Finally, suppose $\nS>\cut$ and $\nE\leq\minind$. Remark that in the proof of Theorem~15 from~\cite{DBLP:conf/fossacs/BerardBLS20}, $\hat N$ increases as $k_{e}$ increases. This means that, since \Sys has a winning strategy with more than \cut processes for \formule when \Envi has \minind processes if and if \Sys has an $(\cut,\minind)$-winning strategy for $\formule$, then \emph{a fortiori} \Sys has a winning strategy with more than \cut processes for \formule when \Envi has \nE processes if and if \Sys has an $(\cut,\nE)$-winning strategy for $\formule$.
  \end{itemize}

  In all of these cases, the search for a winning strategy for $\formule$ can be limited to $[0,\cut]\times[0,\minind]$, which concludes the proof.
\end{proof}

%% file: undecidability.tex
When \Sys and \Envi processes are partitioned, we have seen that when one is only allowed to check whether two positions belong to the same process, the synthesis problem is decidable for \FO. In this section, we show that as soon as we are able to compare the relative positions of two processes, this is no longer the case, even when restricting ourselves to the two-variable setting, and when having access only to one positional relations ($<$ or \plusone): both \PARTSYNTH[\FOtwoeqord] (Theorem~\ref{th:twoeqord}) and \PARTSYNTH[\FOtwoeqplusone] (Theorem~\ref{th:twoeqplusone}) are undecidable.

When processes are shared, the prospect is even darker: it was shown in~\cite{DBLP:conf/fossacs/BerardBLS20} that the synthesis problem in undecidable for $\FOeq$. We argue that this is already the case in the two-variable fragment: \SHAREDSYNTH[\FOtwoeq] is already undecidable (Theorem~\ref{th:twomixed}).

\subsection{Undecidability of \mdseries{\PARTSYNTH[\FOtwoeqord]} and \mdseries{\PARTSYNTH[\FOtwoeqplusone]}}

Let us start by considering the case where the only positional relation is the order. We show in the following that when \Sys and \Envi processes are separated, the synthesis problem is undecidable in this setting.

\begin{theorem}
  \label{th:twoeqord}
  \PARTSYNTH[\FOtwoeqord] is undecidable.
\end{theorem}

\begin{proof}
  We prove this theorem by reduction from the halting problem for two-counter Minsky machines.
  A \emph{two-counter Minsky machine} \mm has a finite set of states \States (containing an initial state \istate and an halting state \hstate), two counters \cntzero and \cntone and a set of transitions \Trans, which is partitioned into \[\Trans=\biguplus_{i=0,1}\Transinc\uplus\Transdec\uplus\Transzero\,,\] where $\Transinc\subseteq\States\times\States$ is the set of transitions incrementing counter \cnt, $\Transdec\subseteq\States\times\States$ is the set of transitions decrementing counter \cnt, and $\Transzero$ is the set of zero-test transitions on counter \cnt. A \emph{configuration} of \mm is a triple $(\state,v_0,v_1)$, where $\state\in\States$ indicates the current state of the machine and $v_i\in\N$ the value of counter $\cnt$. A \emph{run} of \mm is a sequence starting in the initial configuration $(\istate,0,0)$, and such that two successive configurations $(\state,v_0,v_1)$ and $(\statebis,v'_0,v'_1)$ satisfy the following condition: either $v'_i=v'_i+1$, $v'_{1-i}=v_{1-i}$ and $(\state,\statebis)\in\Transinc$, or $v'_i=v'_i-1$, $v'_{1-i}=v_{1-i}$ and $(\state,\statebis)\in\Transdec$, or $v'_i=v_i=0$, $v'_{1-i}=v_{1-i}$ and $(\state,\statebis)\in\Transzero$. A run is \emph{halting} if it is finite and ends in some configuration $(\hstate,v_0,v_1)$ for any $v_i\in\N$. It is undecidable, on input \mm, to tell whether such a halting run exists \cite{minsky1967computation}.

  Given such a machine \mm, we exhibit a formula \formmm, computable from \mm, with the following property: there exists a halting run for \mm iff \Sys has an $(\nS,1)$-winning strategy for \formmm for some $\nS\in\N$. This proves the Theorem.

  We consider the signature $\sig:=\sigE\uplus\sigS$ where $\sigE:=\{\oke,\koe\}$ and \[\sigS:=\{\inc[0],\dec[0],\inc[1],\dec[1],\noop,\oks,\kos\}\uplus\States\uplus\Trans\,.\]

  Let us start by giving an example of a data word encoding a halting run. Suppose that $\States:=\{\istate,\stateone,\statetwo,\hstate\}$ and $\Trans:=\{\transzero,\transone,\transtwo\}$ where
  \[
  \begin{cases}
    \transition[\transzero][\istate][\istate][\cntzero++]\\
    \transition[\transone][\istate][\stateone][\cntzero--]\\
    \transition[\transtwo][\stateone][\statetwo][\cntzero--]\\
    \transition[\transthree][\statetwo][\hstate][\cntzero==0]
  \end{cases}
  \]
  The halting run ($\transzero\cdot\transzero\cdot\transone\cdot\transtwo\cdot\transthree$) of \mm could be represented as the following data word, where we denote \sstyle{\Sys}'s processes by integers, and \estyle{\Envi}'s process as \eproc. Here, \Sys plays mainly on process $0$; they could have mixed their play, as the only time process identity matters is when playing $\inc$ or $\dec$.
  \[
  \begin{aligned}
    \sstyle{(0,\oks)}\estyle{(\eproc,\oke)}&\sstyle{(0,\istate)(0,\transzero)(0,\inc[0])(0,\oks)}\estyle{(\eproc,\oke)}\\
    &\sstyle{(0,\istate)(0,\transzero)(1,\inc[0])(0,\oks)}\estyle{(\eproc,\oke)}\\
    &\sstyle{(0,\istate)(0,\transone)(0,\dec[0])(0,\oks)}\estyle{(\eproc,\oke)}\\
    &\sstyle{(0,\stateone)(0,\transtwo)(1,\dec[0])(0,\oks)}\estyle{(\eproc,\oke)}\\
    &\sstyle{(0,\statetwo)(0,\transthree)(0,\noop)(0,\oks)}\estyle{(\eproc,\oke)}\sstyle{(0,\hstate)}
  \end{aligned}
  \]

  Let us inspect this data word step-by-step. The first two letters of the data word representing a valid run are always an \oks by \Sys followed by an \oke by \Envi.

  Following these two letters, we get a succession of the pattern \[\sstyle{(\_,\state)(\_,\trans)(\_,l)(\_,\oks)}\estyle{(\_,\oke)}\]
  where $\state\in\States, \trans\in\Trans$ and $l$ is either \noop, an \inc or a \dec. Eventually, the data words stops in the halting state \hstate.

  Notice how the value of \cntzero is encoded during this run: at any point during the run, the value of \cntzero
  is equal to the number of \Sys processes on which an \inc[0] has been played, but no \dec[0]. Thus, following a transition in \Transinc, \formmm will force \Sys to play an \inc on a new process. Similarly, after each transition in \Transdec, in order not to lose immediately, \Sys will be forced to play a \dec on a process on which an \inc has been played, but no \dec yet. When the transition is in \Transzero, \Sys must play a \noop on any process; and \formmm grants \Envi a immediate win if \cnt is not zero - that is, if there exists a process on which an \inc and no \dec have been played.

  Let us now explain how players are prevented from cheating to their advantage. The mechanism of fraud prevention is always the same: if \Sys cheats, then \Envi immediately responds by playing \koe, and on the other way around, \Sys plays \oks when detecting a fraud from \Envi. Note that there is by definition an asymmetry between the players, and that the \kos can be differed by the \Envi; but the fairness assumption guarantees that \Sys will be able to output their \kos after some time. Once a \koe or \kos has been played, both players are encouraged to stop (a player who plays after a \ko loses the game), and \formmm checks whether the \ko is justified. If the other player indeed was cheating, then the player who \ko'ed wins, otherwise they lose.

  \Envi can cheat in only three ways:
  \begin{enuE}
  \item\label{enu:e_prefix} by not respecting the prefix $(\_,\oks)(\_,\oke)$
  \item\label{enu:e_multi_ok} by playing multiple \oke in a row
  \item\label{enu:e_early_ok} by playing an \oke before their turn, i.e. before \Sys has played their \oks
  \end{enuE}

  Let us now described the ways \Sys can cheat:
  \begin{enuS}
  \item\label{enu:s_prefix} by not respecting the prefix $(\_,\oks)(\_,\oke)$
  \item\label{enu:s_order} by not respecting the order or number of their letters in a pattern ($\state\to\trans\to\inc/\dec/\noop\to\oks$)
  \item\label{enu:s_target_state} by playing some $\state\in\States$ which is not the end-state of the last transition (or $\istate$ if it is the first state)
  \item\label{enu:s_source_state} by playing some $\trans\in\Trans$ which does not start in the previous state
  \item\label{enu:s_upkeep} by playing an \inc, \dec or \noop which does not match the previous transition
  \item\label{enu:s_double_inc} by playing \inc on a process on which an \inc has already been played
  \item\label{enu:s_double_dec} by playing \dec on a process on which an \dec has already been played
  \item\label{enu:s_dec_no_inc} by playing \dec on a process on which no \inc has been played
  \item\label{enu:s_zero} by playing some $\trans\in\Transzero$ while \cnt is non-zero
  \end{enuS}

  On top of that, \formmm ensures a player loses the game if they refuse to play when their turn comes.
  
  The following notations will be useful:
  \[
  \begin{cases}
    \isstate&:=\quad\bigvee_{\state\in\States}\ \state(x)\\
    \istrans&:=\quad\bigvee_{\trans\in\Trans}\ \trans(x)\\
    \istransinc&:=\quad\bigvee_{\trans\in\Transinc}\ \trans(x)\\
    \istransdec&:=\quad\bigvee_{\trans\in\Transdec}\ \trans(x)\\
    \istranszero&:=\quad\bigvee_{\trans\in\Transzero}\ \trans(x)\\
    \isupkeep&:=\quad\noop(x)\ \lor\ \bigvee_{i=0,1}\ \inc(x)\lor\dec(x)\\
    \issys&:=\quad\bigvee_{a\in\sigS}a(x)\\
    \isenv&:=\quad\bigvee_{a\in\sigE}a(x)
  \end{cases}
  \]

  The following formula \formkos is satisfied when there is an irregularity on \Envi's part since \Sys last played an \oks (either \ref{enu:e_multi_ok} or \ref{enu:e_early_ok}). Thus, \Sys wins if they play \kos while \formkos is satified.

  \[
  \begin{aligned}
    \formkos:=\exists x,&\ \oks(x)\land\big(\forall y,x<y\to\neg\oks(y)\big) \\
    &\land\exists y, x<y\land\oke(y)\land\Big[\big(\exists x, y<x\land\oke(y)\big)\lor \neg\big(\exists y,x<y\land\oks(y)\big)\Big]
  \end{aligned}  
  \]

  Conversely, \formkoe holds when \Sys has cheated since \Envi last played an \oke. Thus, \formmm will ensure that \Envi wins if they output a \koe while \formkoe holds. Due to the number of ways for \Sys to cheat, we introduce subformulas to cover each case.

  \[
  \begin{aligned}
    \formkoe:=\exists x,&\ \oke(x)\land\big(\forall y,x<y\to\neg\oke(y)\big)\\
    &\land\Big[\formbadseq\lor\formbadtarget\lor\formbadsource\lor\formbadupkeep\lor\formbadzerotest\Big]
  \end{aligned}
  \]
  where
  \begin{itemize}
  \item \formbadseq covers \ref{enu:s_order}
  \item \formbadtarget covers \ref{enu:s_target_state}
  \item \formbadsource covers \ref{enu:s_source_state}
  \item \formbadupkeep covers \ref{enu:s_upkeep}, \ref{enu:s_double_inc}, \ref{enu:s_double_dec} and \ref{enu:s_dec_no_inc}
  \item \formbadzerotest covers \ref{enu:s_zero}.
  \end{itemize}

  \[
  \begin{aligned}
    \formbadseq:=&\quad\exists y>x,\ 
    \begin{aligned}[t]
      \Big[\quad&\isstate[y]\ \land\ \exists x>y,\isstate\\
        \lor\ &\istrans[y]\ \land\ \exists x>y,\big(\isstate\lor\istrans\big)\\
        \lor\ &\isupkeep[y]\ \land\ \exists x>y,\big(\isstate\lor\istrans\lor\isupkeep\big)\\
        \lor\ &\oks(y)\ \land\ \exists x>y,\issys\quad\Big]        
    \end{aligned}\\
    &\lor\ \big(\exists y>x,\ \issys[y]\big)\ \land\ \neg\big(\exists y>x,\ \isstate[y]\big)\\ 
    &\lor\ \big(\exists y>x,\ \isupkeep[y]\ \lor\ \oks(y)\big)\ \land\ \neg\big(\exists y>x,\ \istrans[y]\big)\\
    &\lor\ \big(\exists y>x,\ \oks(y)\big)\ \land\ \neg\big(\exists y>x,\ \isupkeep[y]\big)\\
    \formbadtarget:=&\ \neg\big(\exists y>x,\ \istrans[y]\big)\ \land\ \bigvee_{\state\in\States}\Big[\big(\exists y>x,\ \state(y)\big)\ \land\ \exists x,\big([\forall y>x,\ \neg\istrans[y]]\ \land\ \bigvee_{\substack{\trans\in\Trans\text{ doesn't}\\\text{end in }\state}}\trans(x)\big)\Big]\\
    &\lor\ \neg\big(\exists y,\ \istrans[y]\big)\ \land\ \bigvee_{\state\in\States\setminus\{\istate\}}\exists y,\ \state(y)\\
    \formbadsource:=&\bigvee_{\state\in\States}\quad\bigvee_{\substack{\trans\in\Trans\text{ doesn't}\\\text{start in }\state}}\Big[\big(\exists y>x,\ \state(y)\big)\ \land\ \big(\exists y>x,\ \trans(y)\big)\Big]\\
    \formbadupkeep:=&\bigvee_{i=0,1}
    \begin{aligned}[t]
      \Big[\quad&\big(\exists y>x,\ \istransinc[y]\big)\ \land\ \big(\exists y>x,\ \isupkeep[y]\land\neg\inc(y)\big)\\
        \lor\ &\big(\exists y>x,\ \istransdec[y]\big)\ \land\ \big(\exists y>x,\ \isupkeep[y]\land\neg\dec(y)\big)\\
        \lor\ &\big(\exists y>x,\ \istranszero[y]\big)\ \land\ \big(\exists y>x,\ \isupkeep[y]\land\neg\noop(y)\big)\\
        \lor\ &\exists y>x,\ \big(\inc(y)\ \land\ \exists x<y,\ [x\sim y\ \land\ \inc(x)] \big)\\
        \lor\ &\exists y>x,\ \big(\dec(y)\ \land\ \exists x<y,\ [x\sim y\ \land\ \dec(x)] \big)\\
        \lor\ &\exists y>x,\ \big(\dec(y)\ \land\ \neg(\exists x<y,\ [x\sim y\ \land\ \inc(x)]) \big)\quad\Big]
    \end{aligned}\\
    \formbadzerotest:=&\bigvee_{i=0,1}\Big[\big(\exists y>x,\ \istranszero[y]\big)\ \land\ \exists x,\big[\inc(x)\ \land\ \neg(\exists y,\ y\sim x\ \land\ \dec(x))\big]\Big]
  \end{aligned}
  \]

  It only remains to cover cases \ref{enu:e_prefix} and \ref{enu:s_prefix}, with \formprefixE and \formprefixS, as well as the case where a players refuses to play when their turn comes, covered by \formblockE and \formblockS, and the immediate loss when a players keeps playing after a \ko, with \formplayafterkoE and \formplayafterkoS. We use \first, \second and \last to respectively denote the first, second and last element wrt. $<$, which are obviously definable in \FOtwoeqord.

  \[
  \begin{aligned}
    \formprefixE:=&\quad\isenv[\first]\ \lor\ \big[\isenv[\second]\land\neg\oke(\second)\big]\\
    \formprefixS:=&\quad\big[\issys[\first]\land\neg\oks(\first)\big]\ \lor\ \issys[\second]\\
    \formblockE:=&\quad\oks(\last)\\
    \formblockS:=&\quad(\neg\exists x,\top)\ \lor\ \koe(\last)\ \lor\ \big[\isstate[\last]\land\neg\hstate(\last)\big]\ \lor\ \istrans[\last]\ \lor\ \isupkeep[\last]\\
    \formplayafterkoE:=&\quad\isenv[\last]\ \land\ \exists x, x\neq\last\land \big[\koe(x)\lor\kos(x)\big]\\
    \formplayafterkoS:=&\quad\issys[\last]\ \land\ \exists x, x\neq\last\land \big[\koe(x)\lor\kos(x)\big]\\\\
  \end{aligned}
  \]

  We are now ready to make \formmm explicit. We want to \formmm hold either when \Envi has made a misplay (either by satisfying \formprefixE, \formblockE or \formplayafterkoE, or by making a move after which \Sys can play \kos and satisfy \formkos), or if \Sys has made no misplay (in which case the data word represents a valid run) and this run ends in \hstate:
  \[
  \begin{aligned}
    \formmm:=&\quad\formprefixE\ \lor\ \formblockE\ \lor\ \formplayafterkoE\ \lor\ \big[\exists x,\kos(x)\land\formkos\big]\\
    &\lor\Big(\neg\formprefixS\ \land\ \neg\formblockS\ \land\ \neg\formplayafterkoS\ \land\ \neg\big[\exists x,\koe(x)\land\formkoe\big]\ \land\ \hstate(\last)\Big)
  \end{aligned}
  \]

  It remains to explain why \Sys having a winning strategy for \formmm amounts exactly to the existence of a halting run in \mm.
  First, consider a halting run of \mm. It is straighforward to see that \Sys has an $(\nS,1)$-winning strategy for \formmm, where \nS is the total number of transition of $\bigcup_{i=0,1}\Transinc$ used in the run (or $\nS=1$ if no such transitions are used): start by playing \oks on any process, and wait for \Envi to play \oke on their process. Once this is done, play according to the run, in the sequence $\state\to\trans\to\inc/\dec/\noop\to\oks$ (taking care to play \inc on a new process each time, and \dec on a process which has an \inc but no \dec), and after each such sequence wait for the acknowledgement from \Envi, in the form of an \oke. Then stop when reaching \hstate. If \Envi blocks, block as well; if \Envi deviates by playing \oke too soon, play \kos as soon as possible and then block; if \Envi plays \koe, block. It is easily shown by induction on the length of the run that such a strategy is winning for \formmm, and indeed require only $\nS$ processes.

  Conversely, suppose that \Sys has an $(\nS,1)$-winning strategy for \formmm, for some $\nS\in\N$. We argue this entails the existence of a halting run in \mm. Let us inspect the data word \dw produced by this winning strategy against an \Envi which respects the aforementioned rules (namely, which waits their turn to play \oke, unless \Sys cheats, i.e. if the current data word satisfies \formkoe, in which case \Envi plays a \koe and stops). Since \Envi respects the format and makes no false allegation, \formmm cannot be satisfied because of the subformulas \formprefixE, \formblockE, \formplayafterkoE or $\exists x,\kos(x)\land\formkos$. Hence, the second part of \formmm must hold: \formprefixS, \formblockS and  \formplayafterkoS are false, no \koe was played, and $\hstate(\last)$ holds: \Sys's part of \dw must start with an \oks, followed by a sequence of patterns $\state\to\trans\to\inc/\dec/\noop\to\oks$ interspered with \oke from \Envi, until it reaches a position marked with \hstate. Since \Envi played no \koe, at no point was \formkoe satisfied, which means that one can reproduce in a run of \mm the sequence of states and transitions played by \Sys in \dw: by induction on the number of patterns $\state\to\trans\to\inc/\dec/\noop\to\oks$ played by \Sys, one shows that there is a corresponding run in \mm whose last configuration is $(\state,v_0,v_1)$, where \state is the last letter from \States played by \Sys and $v_i$ is the number of processes from \ProcS marked with one \inc and no \dec. The validity of the transitions from $\Transdec$ and $\Transzero$ comes from the fact \formbadupkeep and \formbadzerotest were never satisfied, for otherwise \Envi would have immediately played a \koe. Considering the whole of \dw, we get a run of \mm ending in \hstate, thus concluding the proof.
  \qed
\end{proof}

Let us now argue that this problem remains undecidable if one has access to the successor relation on positions, rather that to the order itself:

\begin{theorem}
  \label{th:twoeqplusone}
  \PARTSYNTH[\FOtwoeqplusone] is undecidable.
\end{theorem}

First of all, note that Theorems~\ref{th:twoeqord} and \ref{th:twoeqplusone} are not derivable from one another, as \FOtwoeqord and \FOtwoeqplusone have orthogonal expressive power: with only two variables, the successor relation is not definable from the order.

\begin{proof}[Sketch of proof]
  The main idea is to note that in the proof of Theorem~\ref{th:twoeqord}, the encoding of a run follows a sequence of patterns of size 5 ($\state\to\trans\to\inc/\dec/\noop\to\oks\to\oke$). Furthermore, as soon as an \oke or \oks is issued, the suspected violation can be found in the positions immediately preceding it (in the case of \oke) or is easily seen globally (in the case of an \oks issued to report an instance of \ref{enu:e_multi_ok} or \ref{enu:e_early_ok}, which can be delayed for arbitraily long by \Envi). As a consequence of this locality, all the formulas in the proof can be reformulated with \plusone instead of $<$.
  \qed
\end{proof}

\subsection{Undecidability of \mdseries{\SHAREDSYNTH[\FOtwoeq]}}\label{sec:undec-synth-mixed}

Let us now turn to the case where processes are shared between both players. In this context, not much can be done: \mdseries{\SHAREDSYNTH[\FOeq]} was shown to be undecidable in \cite{DBLP:conf/fossacs/BerardBLS20}.
As we now show, this problem is already undecidable when we restrict ourselves to two variables.

\begin{theorem}
  \label{th:twomixed}
  \SHAREDSYNTH[\FOtwoeq] is undecidable.
\end{theorem}

\begin{proof}[Sketch of proof]
We give a guideline to adapt the proof of undecidability of \SHAREDSYNTH[\FOeq] (\cite{DBLP:conf/fossacs/BerardBLS20}, Theorem 17) to the case where only two variables are available.
That proof was a reduction from the halting problem for two-counter Minsky machines, but in a way that is quite different from the reduction from the proof of Theorem~\ref{th:twoeqord}: the idea was to encode the value of each counter as the number of processes having a particular number of actions played by both System and Environment.
Namely, to increment a counter, System and Environment would cooperate to pick a fresh process and perform two $a$'s (from System) and two $b$'s (from Environment) on it.
To decrement the same counter, one such process would further receive two additional $a$'s and $b$'s, and from this point on this particular process would never be involved again.
A limited kind of alternation enforced by the formula made sure System and Environment followed this construction; otherwise they immediately lost the game.

The winning condition in the previous proof was given as a condition on the locations of the corresponding parameterized vector game, but it is straightforward to go back to \FOeq formulas from those.
The number of variables needed depends on the largest value appearing in the count of the number of letters for a given process or in the count of the number of processes for a given location.
For instance, we need two different variables to express that a process has at least two $a$'s, or that it has exactly one $a$ (by stating it has at least one and not at least two): \[\varphi_{|a|\geq 2}(x) := \exists y,\ y\sim x\ \land\ a(y)\ \land\ \exists x,\ x\sim y\ \land x\neq y\ \land\ a(x)\,.\]
Similarly if we want to specify that there are at least two such processes then we need at least two different variables.

The counting of processes for a given location is not a problem for the proof in the case of \FOtwoeq, as we only check for the existence of exactly zero (one variable needed), exactly one (two variables needed), at least zero (no variable needed), and at least one (one variable needed) processes in a particular location.
The only hurdle lies in the counting of actions, as we need to count up to four occurrences for each letter (as seen above, four $a$'s and four $b$'s).
The trick to overcome this difficulty is to use multiple letters to play the same role, instead of just one: $a$ is split into $a_1$ and $a_2$, and $b$ into $b_1$ and $b_2$.
And indeed, although it was not possible in \FOtwo to express the existence of a process with four $a$'s or more, we can now use a formula stating the existence of a process with at least two $a_1$'s and two $a_2$'s with only two variables: $\varphi_{|a_1|\geq 2}(x)\land\varphi_{|a_2|\geq 2}(x)$. 
The incrementation of a counter is now encoded by playing exactly one $a_1$, one $a_2$, one $b_1$, and one $b_2$ (instead of exactly two $a$'s and two $b$'s) on a new process, and to decrement a counter, we double the count of each letter as was done previously.
We now only need to count the number of occurrences of each letter up to two, which can be done with two variables instead of four.
The rest of the proof is unchanged.
\end{proof}

%% file: conclusion.tex
In this paper, we have answered the questions left open in \cite{DBLP:conf/fossacs/BerardBLS20}. It appears that when positions between two processes can be compared, the synthesis problem quickly becomes undecidable. As a next step, it thus seems natural to consider the case of partitioned processes for an intermediate logic between \FOtwoeq and \FOtwoeqord: $\FOtwo[\lesssim]$, where one can compare only positions pertaining to the same process.